\documentclass[a4paper,11pt]{article}
\usepackage{underscore} 

\usepackage[utf8]{inputenc} 
\usepackage[english]{babel} 
\usepackage{mathtools}
\usepackage{amsmath}
\usepackage{enumitem}

\usepackage{ebproof}

\usepackage[T1]{fontenc} 
\usepackage{lmodern} 
\usepackage{amssymb} 
\usepackage{stmaryrd} 
\usepackage{cmll}

\usepackage{tikz}
\usepackage{tikz-cd}
\usepackage{tkz-graph}
\usetikzlibrary{matrix,arrows,decorations.pathmorphing}

\usepackage{comment}
\usepackage{amsthm}
\newtheorem{theorem}{Theorem}[section]
\newtheorem{lemma}[theorem]{Lemma}
\newtheorem{proposition}[theorem]{Proposition}

\newtheorem{definition}[theorem] {Definition}
\newtheorem{example}[theorem]{Example}

\newcommand{\redp}{\rightarrow_{\partial}}

\newcommand*{\N}{\mathbb N}

\newcommand*{\recdef}{\Coloneqq}
\newcommand*\sums[1]{\N[#1]}
\newcommand*\cons{\dblcolon}
\newcommand*\length[1]{\lvert #1\rvert}

\newcommand*\rappl[3][]{#1\langle #2#1\rangle #3}
\newcommand*\subst[4][]{#2 #1[#4/#3 #1]}
\newcommand*\nsubst[3]{\partial_{#2} #1\cdot #3}

\newcommand*\ndTerms{\Lambda}
\newcommand*\resTerms{\Delta}
\newcommand*\resMonomials{\resTerms^\oc}
\newcommand*\resExprs{\resTerms^{(\oc)}}

\newcommand*\rigidTerms{D}
\newcommand*\rigidMonomials{\rigidTerms^\oc}
\newcommand*\rigidExprs{\rigidTerms^{(\oc)}}
\newcommand*\pare{\Rightarrow_{\partial \epsilon}^{\neg e}}
\newcommand*{\nsubstaa}[4]{\sum\limits_{(I'_{0}, I'_{1}) \text{ partition of }  \{1,\dotsc,n\} } \rappl { \nsubst  {#1}{#3} {\bar{#4}_{I'_{0}}} } {\nsubst  {\bar{#2}} {#3} {\bar{#4}_{I'_{1}} }}}
\newcommand*{\nsubsta}[4]{\sum\limits_{(I_{0}, I_{1}) \text{ partition of }  \{1,\dotsc,n\} } \rappl { \nsubst  {#1}{#3} {\bar{#4}_{I_{0}}} } {\nsubst  {\bar{#2}} {#3} {\bar{#4}_{I_{1}} }}}
\newcommand*{\rexp}[1]{T_{r}(#1)}

\newcommand*{\NF}[1]{NF(#1)}
\newcommand*{\tm}{a}
\newcommand*{\tmtwo}{b}
\newcommand*{\tmthree}{c}
\newcommand*{\tmfour}{d}
\newcommand*{\Hr}[2]{H_r^{#2} (#1)}

\newcommand*{\rhred}{\lambda x_1 \dots \lambda x_m . \rappl{\cdots \rappl{\rappl{\lambda x. p}{\bar{q_0}}}{\bar{q_1}\cdots}} { \bar{q_n}}}

\newcommand*{\tope}{\to_{\partial \epsilon}}

\newcommand*\definitive[1]{\emph{#1}}

\title{Normalization, Taylor expansion and rigid approximation of $\lambda$-terms }
\author{Federico Olimpieri}

\date{}

\begin{document}

\maketitle

\tableofcontents

\abstract{The aim of this work is to characterize three fundamental normalization proprieties in lambda-calculus trough the Taylor expansion of $ \lambda$-terms. The general proof strategy consists in stating the dependence of ordinary reduction strategies on their resource counterparts and in finding a convenient resource term in the Taylor expansion that behaves well under the considered kind of reduction.}

\section{Introduction}

The Taylor expansion of ordinary $\lambda$-terms has been introduced in \cite{er:tay} as a syntactic counterpart of the quantitative semantics of linear logic in a $\lambda$-calculus setting. Thanks to this semantic work, 
 Ehrhard and Regnier \cite{er:diff} were able to present, in a very natural way, a differential extension of $\lambda$-calculus. 
 The derivative of a term with respect to its argument is, following the classical analytical notion, a ``linear approximation'' of it. In this case, linearity has a logical meaning: 
 variables in the derivatives are used only once during the reduction process. In \cite{er:tay} they presented a fragment of this calculus, called resource $\lambda$-calculus, where one considers only derivatives of terms applied to zero and
 a notion of Taylor expansion of ordinary $\lambda$-terms can be introduced.

The aim of this work is to characterize three fundamental normalization proprieties in $\lambda$-calculus trough the Taylor expansion. More precisely, we shall introduce a rigid version of the resource calculus, replacing multisets with lists following \cite{mazza:pol} and \cite{tao:gen}. An element of the Taylor expansion can then be seen as an equivalent class of rigid resource terms. The general proof strategy will consists in stating the dependence of ordinary reduction strategies for $\lambda$-calculus on their rigid resource counterparts and in finding a convenient rigid approximant that behaves well under the considered kind of reduction. The choice of rigid terms over standard resource terms remarkably simplifies definitions, theorems and proofs. Moreover, in section 6, we establish the relationship between the rigid expansion and the standard Taylor expansion of $ \lambda $-terms.

The ideas and methods used in this work derive mostly from intuitions and results presented firstly in \cite{carv:sem} and \cite{er:tay}. The characterization of head-normalization that we shall present has been folklore for some time. An important ispiration is \cite{guerr:lam}, where solvability \textit{via} Taylor expansion is considered from a call-by value perspective. For what concerns $\beta$-normalization, the result  derives directly from Lemma \ref{commL}, that has been proven firstly in \cite{vaux:alg}, and it is inspired also by \cite{carv:weak}. 

The result about strong normalization is new.  Our characterisation differs substantially from the one given in \cite{vaux:finit}, where the strong normalisation is achived \textit{via} a global proprerty on the Taylor expansion. Instead we focus on an existantial proprerty, namely the non-zero termination of an extended non-erasing reduction (see Section \ref{strong}). The idea of considering non-erasing reduction derives from \cite{carv:strong} and from the $ \lambda I$-calculus (see Section \ref{cons}).  

Our most important contribution is our approach: we give a general method to state these characterization \textit{via} rigid approximation and, \textit{a fortiori}, the Taylor expansion. The strength of our approach is also evident for it produces a completely straightforward proof of normalisation for for the head and left reductions (see Theorems \ref{hnorm} and \ref{bnorm}). This happens thanks to the finitary nature of resource calculus operational semantics.

  Our method can be also straightforwardly extended to prove typability results for (intersection) type systems, without passing trough Girard's candidates of reducibility. We also believe that this approach can be extended to the study of the execution time for $\lambda$-terms, in the sense of \cite{carv:sem} and to prove similar results in the context of Bang Calculus \cite{guerr:bang} \cite{gg:bang} and Multiplicative Exponential Linear Logic.

\section{Rigid resource terms}

We introduce a resource sensitive calculus following \cite{er:tay}. In this calculus the number of copies of the argument that a term uses under reduction is made explicit \textit{via} lists of terms. Following \cite{tao:gen}, we call this calculus the rigid resource calculus. Rigidity means that resources are modelled by lists instead of multisets.\footnote{As it is the case of \cite{er:tay}.} We will denote as "resource term" both rigid resource terms and standard ones. The distinction between the two will be clarified either by the context or explicitly, if needed. 

 We  define the set of \emph{rigid resource terms} $\rigidTerms$
	and the set of \emph{rigid resource monomials} $\rigidMonomials$
	by mutual induction  as follows: 
\[
	\rigidTerms \ni \tm, \tmtwo, \tmthree ::=
		\tm \mid \lambda x. a \mid \rappl{\tmthree}{\vec{\tmfour}} \mid 0
	\qquad\qquad
	\rigidMonomials \ni \vec{\tm}, \vec{\tmtwo}, \vec {\tmfour} ::=
		() \mid (a) :: \vec{\tmfour}
\]
If $A$ is a set, $A^{!}$ denotes the set of lists over $A$.
 Rigid monomials are then lists of resource terms and $ \vec{a} \cdot \vec{b}$ denotes list concatenation. We write $(\tm_1,\dotsc, \tm_n)$ for $(\tm_1)\cdot\ldots\cdot(\tm_n)\cdot()$. A term of the form $\langle \tmthree \rangle \vec{\tmfour}$ is called a linear application. The $ 0 $ term works as a zero linear combination, $ \textit{i.e.}$ $\lambda x. 0 = 0 $, $ \rappl{0}{\vec{\tm}} = 0 $, $ \rappl{\tm}{0} = 0 $ and $ (0)\cdot \vec{\tm} = 0 $.
We call \emph{rigid resource expressions} the elements of $\rigidExprs=\rigidTerms\cup\rigidMonomials$.
For any resource expression $e$, we write $n_x(e)$ for the number of occurrences of variable $x$ in $e$.

 We define the rigid substitution: 
 
 \begin{definition}
	We define $e\{\vec{b}/x\}$ for any $e\in\rigidExprs$ and $\vec b\in\rigidMonomials$ such that $\length{\vec b}=n_x(e)$
	inductively:
\begin{gather*}
	x\{(t)/x\}=b \qquad y\{()/x\} = y 
	\\
	(\lambda y. s)\{\vec{b}/x\}= \lambda y. (s \{\vec{b}/x\} )
	\qquad
	(\langle s \rangle \vec{b})\{\vec b_0\cons\vec b_1/x\}
	=\rappl{s \{\vec{b}_{0}/x\}} \vec{d}\{\vec{b}_{1}/x\}
	\\
	(a_{1},\dotsc, a_{n})\{\vec{b}_{1}\cons \cdots\cons \vec{b}_{n}/x\}
	=( a_{1}\{\vec{b}_{1}/x\},\dotsc, a_{n}\{\vec{b}_{n}/x\})
\end{gather*}
whenever $y\not=x$, $y\notin FV(\vec b)$, 
$\length{\vec b}=n_x(a)$,
$\length{\vec b_0}=n_x(c)$,
$\length{\vec b_1}=n_x(\vec d)$,
and $\length{\vec b_i}=n_x(a_i)$ for $1\le i\le n$.
\end{definition}
 
\begin{definition}
Let $e\in \rigidExprs$, $x \in \mathcal{V}$ and $\vec{b}\in\rigidMonomials$.
We define $a[\vec{b}/x]$ the \definitive{rigid substitution} of $\vec{b}$ for $x$ in $a$,
setting $a[\vec{b}/x] = a\{\vec{b}/x\}$ if $n_{x}(a)=\length{\vec b}$ and $a[\vec{b}/x]=0$ otherwise.
\end{definition}
 
The reduction of rigid resource terms has the following base cases:
\[
\rappl{ \lambda x. a}{\vec{b}} \to_{r} a[\vec{b}/x]
\]
extended contextually.

\begin{proposition}\label{snconf}
The reduction $\to_{r}$ is confluent and strongly normalizing.
\end{proposition}
\begin{proof}
The result follows from the fact that the size of resource terms is decreasing under reduction. See \cite{er:tay}. 
\end{proof}

We write $\NF{a}$ for the unique normal form of $a$ that is a rigid term or $0$.

  \begin{example}
   The rigid resource version of $\Omega$ reduces to $0$:
   \[\langle \lambda x. \langle x \rangle (x) \rangle (\lambda x. \langle x \rangle (x)) \redp 0 \]
   This happens because the number of times that $x$ is called differs from the number of arguments available.
   \end{example}

 Let $M$ be a $\lambda$-term. We inductively define $\rexp{M}\subseteq \rigidTerms,$ the \emph{rigid expansion} of $M$, as follows:
 \begin{itemize}
  \item if $ M = x$ then $ \rexp{M} = \{ x \};$
  \item if $ M = \lambda x. M'$ then $ \rexp{M} = \lambda x. T_{r}(M') = \{ \lambda x. a \ \mid a \in \rexp{M'} \} ;$
  \item if $ M = PQ$ then $ \rexp{M} = \langle \rexp{P} \rangle \rexp{Q}^{!} = \{ \langle c \rangle \vec{d} \ \mid c \in \rexp{P} \\ \text{ and } \vec{d}\in \rexp{Q}^{!} \}.$
 \end{itemize}

\section{Head normalization} \label{head}

The first characterization that we give concerns head-normalization. This result is folklore  but we give a novel presentation of it following our general approach. Firstly we give recall some basic definitions and results.

 \begin{proposition} \label{canonLa}
  Let $M \in \Lambda$. There exist $ x_{1},...,x_{m} \in \mathcal{V} $ and\\ $ M', N_{1},..., N_{n} \in \Lambda$, with $M'$ either a redex 
  or a variable, such that
  $ M = \lambda x_{1}...\lambda x_{m}. M' N_{1}...N_{n}.$
 \end{proposition}
 \begin{proof}
  Trivial induction on the size of $M$.
 \end{proof}

From now on we will use the former proposition as a characterisation of $ \lambda$-terms without explicitly referring to it.

 If $ M'= x$ with $ x\in \mathcal{V}$ we say that $M$ is a \emph{head-normal form}. If $ M'$ is a redex it is called the \emph{head-redex}  of $M$.
 We write $M \to_{h} N$ if $ M = \lambda x_{1}...\lambda x_{m}. (\lambda x. P) N N_{1}...N_{n} $ and
    $N=  M = \lambda x_{1}...\lambda x_{m}. P[N/x] N_{1}...N_{n} $. We say that $M$ is \emph{head-normalizable}  if
  there exist $M_{1},....,M_{n} \in \Lambda$
such that $ M=M_{0} \to_{\beta} M_{1} \to_{\beta} ... \to_{\beta} M_{n}$ with $M_n$ head-normal form. 

 \begin{proposition} \label{canonR}
  Let $\tm \in \rigidTerms$. There exist $ x_{1},...,x_{m} \in \mathcal{V} $ and \\ $ \tm', \vec{\tmtwo}_{1},..., \vec{\tmtwo}_{n} \in \rigidExprs$, with $\tm'$ either a redex 
  or a variable, such that
  \[\tm = \lambda x_{1}...\lambda x_{m}. \rappl{ \cdots \rappl{ \tm'} {\vec{\tmtwo}_{1}}}{ \cdots \vec{\tmtwo}_{n}}\] or $ \tm = 0.$
 \end{proposition}
 \begin{proof}
  Trivial induction on the size of $\tm$.
 \end{proof}

From now on we will use the former proposition as a characterisation of rigid resource terms without explicitly referring to it.

 If $ \tm'= x$ with $ x\in \mathcal{V}$ or $ \tm = 0 $ we say that $\tm$ is a \emph{head-normal form}. If $ \tm'$ is a redex it is called the \emph{head-redex}  of $M$.
 We write $\tm \to_{h} \tmtwo$ if $ \tm = \lambda x_{1}...\lambda x_{m}. \rappl{ \cdots \rappl{\rappl{ \lambda x. \tm'} {\vec{\tmtwo}}} {\vec{\tmtwo}_{1}}}{ \cdots \vec{\tmtwo}_{n}} $ and
    $ \lambda x_{1}...\lambda x_{m}. \rappl{ \cdots \rappl{ \nsubst {\tm'}{x}{\vec{\tmtwo}}} {\vec{\tmtwo}_{1}}}{ \cdots \vec{\tmtwo}_{n}} $. We say that $\tm$ is \emph{head-normalizable}  if
  there exist $\tm_{1},....,\tm_{n} \in \rigidTerms$
such that $ \tm=\tm_{0} \to_{r} \tm_{1} \to_{r} ... \to_{r} \tm_{n} $ with $\tm_n$ head-normal form.

\subsection{Towards head normalization}

The first step is a clear statement of what happens to the rigid expansion under substitution. We set $\rexp{M}[\rexp{N}/x] = \{ a[\vec{b}/x] \mid a \in \rexp{M} \text{ and } \vec{b} \in \rexp{N}^{!} \text{ s.t. }  a[\vec{b}/x] \neq 0 \}. $

\begin{lemma} \label{Taysub} Let $ M $ and $N$ be two $\lambda$-terms. Then \[ \rexp{M[N/x]} = \rexp{M}[\rexp{N}/x]. \]
\end{lemma}
\begin{proof}
By induction on the structure of $ M[N/x] $.

 If $ M = x $ then $ \rexp{M[N/x]} = \rexp{N} $. By definition of rigid substitution we have that $\rexp{M}[\rexp{N}/x] = \{ x[(b)/x] \mid b \in \rexp{N} \} = \rexp{N}. $ If $ M = y $ with $ y \neq x $ we have that $ \rexp{M[N/x]} = \{ y \} $. Then by definition of rigid substitution we can conclude. 
 
 If $ M = \lambda x. M'$ the result derives immediately by IH.
 
  If $ M = PQ $ then $ \rexp{M} = \langle \rexp{P} \rangle \rexp{Q}^{!}$. By definition we have that \[ \subst{\rexp{M}} {x} {\rexp{N}} = \rappl{\subst{\rexp{P}} {x} {\rexp{N}}} {\subst {\rexp{Q}^{!}} {x} {\rexp{N}}}.  \]
  
  and that 
  
  \[ \rexp{\subst{M} {x} {N}} =  \rappl{\rexp{\subst{P} {x} {N}}} {\rexp{\subst {Q} {x} {N}}^{!} } \]

 By IH we have that \[ \rexp{\subst{P} {x} {N}} = \subst{\rexp{P}} {x} {\rexp{N}}\] and that $ \rexp{\subst {Q} {x} {N}} =  \subst{\rexp{Q}} {x} {\rexp{N}}$. Hence \[\rexp{\subst {Q} {x} {N}}^{!} =  \subst{\rexp{Q}} {x} {\rexp{N}}^{!}.\] We can then apply the IH and conclude.

\end{proof}

\begin{lemma} \label{antired1}
Let $M, N$ be any two $\lambda$-terms. If $M \to_{\beta} N$ then for all $ b \in T_{r}(N) $ there exists $a \in T_{r}(M)$ such that $ a \to_{r} b. $

\end{lemma}

\begin{proof}
By induction on the definition of $\beta$-reduction and by Lemma \ref{Taysub}. 
\end{proof}

Since we have a substitution Lemma, the next natural step is asking what happens to the rigid expansion under reduction. Since we are focusing on head-normalization, we can restrict our considerations on the head-reduction.\par
First of all we shall give a functional definition of head-reduction, that will allows us to state in a compact way the connection between head-reduction and resource head-reduction. 

\begin{definition}
Let $M$ be a $\lambda$-term. We define the head-reduction $H: \Lambda \to \Lambda $ by cases as follows: \par
$H(M)=  \begin{cases} M & \text{ if }M \text{ is a head-normal form};\\ 
\lambda x_{1}...\lambda x_{m}. P[Q/x] Q_{1}...Q_{n} & \text{ otherwise. } 
\end{cases} $

\end{definition}

\begin{definition}
Let $s\in \rigidTerms$. We define the head resource reduction $H_{r}: \rigidExprs\to \rigidExprs$ by cases as follows:\par
 $\Hr{\tm}{}=  \begin{cases} \tm & \text{ if } \tm \text{ is a hnf};
 \\
\lambda x_{1}...\lambda x_{m}. \langle... \langle \subst{\tmthree}{x}{\vec{\tmfour}}\rangle \vec{\tmfour}_{1}... \rangle \vec{\tmfour}_{n} & \text{ otherwise. } 
\end{cases} $\\
 
$\Hr{\vec{\tmtwo}}{} = (\Hr{\tmtwo_{1}}{},..., \Hr{\tmtwo_{n}}{})$ with $\vec{\tmtwo}=(\tmtwo_{1},..., \tmtwo_{n})$.

 \end{definition}
 
 We say that the head-reduction of $ M $(resp. $ \tm$) \emph{ends} if there exists $ m \in \mathbb{N} $ such that $ H^{m}(M) $ (resp. $ \Hr {\tm } {m}$) is a head-normal form. In that case we call $ H^{m}(M) $ (resp. $ \Hr {\tm } {m}$) the \emph{principal normal form} of $M $(resp. $ \tm$). We denote the principal normal form of $ M$(resp. $ \tm$) as $ HNF(M)$(resp. $ HNF(\tm)$).

We Set $H_{r}(\rexp{M})= \{ H_{r}(s) \mid \tm \in \rexp{M} $ and $ H_{r}(\tm) \neq 0 \} $. Then:

\begin{lemma} \label{commH}
Let $M$ be a $\lambda$-term. Then $\Hr{\rexp{M}}{}= \rexp{H(M)}.$\end{lemma}
\begin{proof}
 We prove the result by cases and double inclusion.\par
If $M$ is a head-normal form the result is trivial. \par If $M= \lambda x_{1}....\lambda x_{m}. (\lambda x. P)QQ_{1}...Q_{n}$, we can focus on \[M'= (\lambda x. P)QQ_{1}...Q_{n}\] without any loss of generality. Then $ H(M) = P[Q/x]Q_{1}...Q_{n} .$
By Lemma \ref{Taysub}, we have that \[T(H(M')) = \langle ... \langle \subst {\rexp{P}} {x} {\rexp{Q}} \rangle T(Q_{1})^{!}....\rangle T(Q_{n})^{!}\]

Let $s'\in H_{r}(\rexp{M'})$. By definition there exists $\tm\in \rexp{M}$ such that $s'=H_{r}(\tm)$. Since $\tm$ must be of the form $ \langle ... \langle \lambda x. \tmthree \rangle \vec{\tmfour} \rangle \vec{\tmfour}_{1}...\rangle \vec{\tmfour}_{n} $, for some $ p\in T_{r}(P),\vec{\tmfour}\in T(Q)^{!}$ and $ \vec{\tmfour}_{i}\in T_{r}(Q_{i})^{!}$, for $i\in \{1,...,n\}$. By definition of $H_{r}$, $\tm'= \langle ... \langle \subst {\tmthree} {x} {\vec{\tmfour}} \rangle \vec{\tmfour}_{1}... \rangle \vec{\tmfour}_{n} )$. Then $\tm' \in T_{r}(H(M'))$. \par
Conversely, Let $\tm' \in T_{r}(H(M'))$. By definition, \[\tm'= \langle ... \langle \subst {\tmthree} {x} {\vec{\tmfour}} \rangle \vec{\tmfour}_{1}... \rangle \vec{\tmfour}_{n}) \] for some $ \tmthree\in T_{r}(P),\bar{\tmfour}\in T_{r}(Q)^{!}$ and $ \vec{\tmfour}_{i}\in T_{r}(Q_{i})^{!}$, for $i\in \{1,...,n\}$. Then there exists $\tm = \langle ... \langle \lambda x. \tmthree \rangle \vec{\tmfour} \rangle \vec{\tmfour}_{1}... \rangle \vec{\tmfour}_{n} \in T_{r}(M')$ such that $\tm'=  H_{r}(\tm).$ 
\end{proof}

The meaning of the former lemma is that performing a step of head-reduction on $M$ and then computing the rigid expansion of its retract is the same thing as first computing the rigid expansion of $M$ and then performing a step of head-reduction on it. We say then that head-reduction and rigid expansion \emph{commute}.

\begin{lemma} \label{hnf}
Let $\tm \in \rigidTerms$ be a resource head-normal form. Let $\lambda$-term $M$ such that $\tm \in \rexp{M}$, then $M$ is a head-normal form. 
\end{lemma}
\begin{proof}
Let $M= \lambda x_{1}...\lambda x_{m} P Q_{1}...Q_{n}$, with $P$ either a redex or a variable. We can focus on $M' = P Q_{1}...Q_{n}$ without any loss of generality. By definition, $\rexp{M'}= \langle \rexp{P} \rangle \rexp{Q_{1}}^{!} \rangle ... \rangle \rexp{Q_{n}}^{!}$. Since $\tm \in \rexp{M}$, $\tm= \langle x \rangle \bar{\tmfour}_{1} ... \rangle \bar{\tmfour}_{n}$, with $x\in \rexp{P} $ and $ \vec{\tmfour}_{i} \in \rexp{Q_{i}}^{!}$ for $i\in \{1,...,n\}$. Then, by definition of $\rexp{P}$, $P$ must be a variable and hence $M'$ is a head-normal form. 
\end{proof}

\begin{lemma}\label{nftohnf}
Let $ \tm \in \rigidTerms. $ if $ NF(\tm) \neq 0 $ then $ HNF(\tm) \neq 0. $
\end{lemma}
\begin{proof}
By Lemma \ref{snconf}.
\end{proof}

\begin{proposition}\label{imph}
Let $M\in \Lambda$. If there exist a resource term $\tm \in \rexp{M}$ and $ m \in \mathbb{N} $ such that $H_{r}^{m}(\tm) \in D_{HNF} \setminus \{ 0 \} $ then $H^{m}(M)$ is a head-normal form.
\end{proposition}
\begin{proof}
 By Lemma \ref{commH}, $H^{i}(\tm) \in \rexp{H^{i}(M)}$ for $i\in \{1,...,n\}$. Then by Lemma \ref{hnf}, $H^{n}(M)$ is a head-normal form of $M$ (precisely the principal head-normal form of $M$).

\end{proof}

\begin{proposition} \label{phnf}
Let $M \in \ndTerms $. If $ M $ is head-normalizable then there exists $\tm \in \rexp{M}$ such that $NF(\tm) \neq 0$.
\end{proposition}
\begin{proof}
Let $ N $ be a head-normal form of $ M$. By definition a $ \tmtwo_{0} \in T_{r}(N)$ is a resource head-normal form. Also by definition $ M = M_{0} \to_{\beta} \dots \to_{\beta} M_{n} = N. $ Then, by Lemma \ref{antired1}, we can conclude, since we take the rigid term $ \langle \cdots \langle x \rangle () \cdots \rangle () $ and follow its anti-reduction.

\end{proof}

 \begin{theorem}
 \label{hnorm}
Let $M\in \ndTerms$ . the following statements are equivalent:\\
(i) there exists $\tm \in \rexp{M} $ such that $ NF(\tm) \neq 0; $ \\
(ii) there exist a resource term $\tm\in \rexp{M}$ and $ m \in \mathbb{N} $ such that $H_{r}^{m}(\tm)\in D_{HNF} \setminus \{ 0\} ;$\\
(iii)  $ H^{m}(M) $ is a head-normal form;\\
(iv) M is head-normalizable.

\end{theorem}

\begin{proof}
$(i)\Rightarrow (ii) $ is Lemma \ref{nftohnf}. $(ii)\Rightarrow (iii) $ is Proposition $\ref{imph}. $ $(iii) \Rightarrow (iv) $ is is a triviality. Finally, $(iv) \Rightarrow (i)$ is Proposition $ \ref{phnf}$. 
\end{proof}

\subsection{Solvability}

\begin{definition}
Let $ M \in \lambda $ that is closed. Then $ M $ is solvable if there exists $ N_1, \dots, N_n \in \Lambda $ such that $ M N_1 \cdots N_n =_{\beta} \lambda x. x .$ 
\end{definition}

In particular we have that $ M N_1 \cdots N_n \to^{*}_{\beta} \lambda x. x .$ We say that a generic $ M \in \Lambda $ is solvable if there exists a closure of $ M $ that is solvable. 

\begin{theorem}
$ M $ is solvable iff $ M $ is head-normalisable. 
\end{theorem}
\begin{proof}
$ (\Leftarrow)$ If $ M $ is solvable then there exists $ N_1, \dots, N_n \in \Lambda $ such that $ M' N_1 \cdots N_n $ is head-normalisable, with $ M' $ being a closure of $M $. By Theorem \ref{hnorm} we have that there exists $\tm \in \rexp{ M N_1 \cdots N_n} $ such that $\NF{s} \neq 0 .$ We have that $ \tm = \langle \cdots \langle \tmthree \rangle \vec{\tmfour}_1 \cdots \rangle \vec{\tmfour}_n  $ for some $ \tmthree \in \rexp{M} $ and $\vec{\tmfour}_{i} \in \rexp{N_i}^{!} $. By Lemma \ref{snconf} and Theorem \ref{hnorm} we can conclude, since $ \tmthree \in \rexp{M} $ and $ \NF{\tmthree} \neq 0 $. \\
$ (\Rightarrow) $ The same proof of Theorem 8.3.14 of \cite{bar:lam}.
\end{proof}

\section{$\beta$-normalization} \label{secbeta}

In this section we shall present a characterization of $\beta$-normalization  \textit{via} rigid expansion. Firstly we give recall some basic definitions.

 We say that $M\in \Lambda$ is in \emph{$\beta$-normal form}  if 
$M$ does not contain redexes as sub-terms. We say that a term $ M$ is \emph{$\beta$-normalizable} if is there exist $M_{1},....,M_{n}, N\in \Lambda$
such that $ M = M_{0} \to_{\beta} M_{1} \to_{\beta} ... \to_{\beta} M_{n} \to_{\beta} M_{n} = N$ with $N$ $\beta$-normal form. We call $ M= M=M_{0} \to_{\beta} M_{1} \to_{\beta} ... \to_{\beta} M_{n} \to_{\beta} M_{n} = N$ a $\beta$-reduction chain starting from $M$. We say that $M$ is strongly normalizable if there is no infinite $\beta$-reduction chain starting from $M$. 
We extend the former definitions to resource terms in the natural way.

\subsection{$\beta$-normalization \textit{via} rigid expansion}

In order to achieve $ \beta$-normalisation, we will introduce  a parallel version of the left reduction:

 \begin{definition}
Let $M\in \ndTerms$. We define the left-parallel reduction  $L : \Lambda \to \Lambda$ by cases as follows: \par
$L(M)=  \begin{cases} 
M  &  \text{ if M is a normal form} \\
\lambda x_{1}...\lambda x_{m}. x L(Q_{1})...L(Q_{n}) & \text{ if M is a head-normal form} ;\\
\lambda x_{1}...\lambda x_{m}. P[Q/x] Q_{1}...Q_{n} & \text{ otherwise. } 
\end{cases} $
 \end{definition}
 
and the resource version of left-parallel reduction: 
 
 \begin{definition}
Let $\tm\in \rigidExprs$. We define the left-parallel resource reduction $L_{\partial}: \rigidExprs \to \rigidExprs$ by cases as follows:\par
 $L(\tm)_{r}=  \begin{cases}
\tm &  \text{ if \tm is a normal form} \\
   \lambda x_{1}...\lambda x_{m}. x L(\vec{\tmfour}_{1})...L(\vec{\tmfour}_{n}) & \text{ if \tm is a head-normal form} ;\\
\lambda x_{1}...\lambda x_{m}. \langle... \langle \subst {\tmthree} {x} {\vec{\tmfour}} \rangle \vec{\tmfour}_{1}... \rangle \vec{\tmfour}_{n} & \text{ otherwise. } 
\end{cases} $\\
 
$L_{r}(\vec{\tmfour}) = (L_{r}(\tmfour_{1}),..., L_{r}(\tmfour_{n}))$ with $\vec{\tmfour}=(\tmfour_{1},..., \tmfour_{n})$.

 \end{definition}
 
  When $M$ is a head-normal form, then the reduction is propagated to the arguments $Q_{1},..., Q_{n}$ and it is possible to reduces more then one redex for step. If the left-parallel reduction ends then $M$ is $\beta$-normalizable:   
 
 \begin{proposition}
 Let $M$ be a $\lambda$-term. If there exists $n\in \mathbb{N}$ such that $L(M)^{n} = N$ with $N $ $beta$-normal form, $M$ is $\beta$-normalizable. 
 \end{proposition}
\begin{proof}
The proof is trivial. Since, by definition, $ M \twoheadrightarrow^{+}_{\beta} L(M)$,  if there exists such  $n \in \mathbb{N}$ we have a $\beta$-reduction chain starting from $M$ that ends with a $\beta$-normal form, i.e. $M$ is $\beta$-normalizable.
\end{proof}

We set  $L_{r} (T(M))= \{ L_{r}(\tm) \mid \tm \in \rexp{M} $ and $ L_{r}(\tm) \neq 0 \}$. 
Then We state a result that extends Lemma \ref{commH}:

 \begin{lemma}\label{commL} Let $M$ be a $\lambda$-term. Then $ L_{r}(T(M)) = T(L(M))$.

\end{lemma}
\begin{proof}

We prove the lemma by induction on the definition of $L$.

 If $M$ has a head redex then the result follows immediately from Lemma \ref{commH}. 
 
 If $M= \lambda x_{1}... \lambda x_{m}. x  Q_{1}... Q_{n}$ we can focus on $M'=x  Q_{1}... Q_{n}$ without any loss of generality. 
 
  We have that \[T_{r}(L(M')) = \langle ... \langle \{x\} \rangle T_{r}(L(Q_{1})^{!} ... \rangle T_{r}(L(Q_{n}))^{!}.\] By IH we have that
  for all $ i\in \{1,...,n \}$, $ L_{r} (T_{r}(Q_{i})) = T_{r}(L(Q_{i}))$. 
  It remains to prove that, if for all $ i\in \{1,...,n \}$, $ L_{r} (T_{r}(Q_{i})) = T_{r}(L(Q_{i}))$, then $ (T_{r}(M')) = T_{r}(L(M'))$.  Let $\tm\in T_{r}(L(M'))$, then $\tm = \langle ... \langle x \rangle \vec{\tmfour}_{1} \rangle...\rangle \vec{\tmfour}_{n}$ with $\vec{\tmfour}_{i}\in T_{r}(L(Q_{i})) = L_{r}(T_{r}(Q_{i}))$ by hypothesis. Hence, by definition, there exists $ \vec{\tmfour}'_{i}$ for all $i \in \{1,...,n\}$ such that $ \vec{\tmfour}_{i} = L_{r}(\vec{\tmfour}'_{i})$. Let $s' = \langle ... \langle x \rangle \vec{\tmfour}'_{1}\rangle ... \rangle \vec{\tmfour}'_{n}$. Then $ \tm = L_{r} (\tm') \in L_{r} (T(M')).$ The other inclusion is trivial by definition. 
\end{proof}

We observe that a non $\beta $-normalisable term can have a non zero rigid expansion: take $M = x \Omega $ and $ \rappl{x}{()} \in \rexp{M}$. $  \rappl{x}{()} $ is in normal form, but $ M $ is only head-normalisable. The we need to strengthen our hypothesis on the rigid expansion.

\begin{definition}
We inductively define the set of positive resource terms  $ D^{+}\subset D$  as follows: 

\[ s\in \rigidTerms^{+} ::=  x \mid  \lambda x. \tm \mid \langle \tmthree \rangle \vec{\tmfour} \]
\par
with $\vec{q}\in \rigidMonomials$ such that $\vec{q}\neq ()$. We denote $\rigidTerms_{NF^{+}}$ the set of resource normal forms that are also positive resource terms.
\end{definition}

A positive resource term is a resource term where all the arguments are defined: there are no empty lists appearing as arguments in a linear application. Then an approximant of this kind presents all the information about its corresponding $\lambda$-term.

\begin{lemma} \label{forcingL}
Let $M \in \ndTerms$. If there exists $s\in \rexp{M}$ such that $s\in \rigidTerms^{+}$ and $s$ is a resource normal form then $M$ is a $\beta$-normal form.
\end{lemma}
\begin{proof}
By induction on $M$. The non-trivial case is the application. If $M = PQ$ we have that $s= \langle \tmthree \rangle \vec{\tmfour}$ for some $\tmthree\in \rexp{P}$ and $\vec{q}\in \rexp{Q}^{!}$. Since $\tm$  is a positive term, $\vec{\tmfour} = (\tmfour_{1},\dots, \tmfour_{n}) $ for some $ \tmfour_{1},\dots, \tmfour_{n}\in \rexp{Q}$. Since $\tm$ is a resource normal form, we have that $\tmthree$ is a positive resource normal form and, for $i\in \{1,...,n\}$ $\tmfour_{i}$ is a positive resource normal form. Hence we can apply the IH and conclude.
\end{proof}

\begin{lemma} \label{standL}
Let $\tm\in \rigidTerms$. there exists $n\in \mathbb{N}$ such that $ L_{r}^{n}(\tm) = NF(\tm)$.
\end{lemma}
\begin{proof}
Trivial by Lemma \ref{strong} and by the observation that $s\twoheadrightarrow^{+} L_{r}(\tm)$.
\end{proof}

Since the left-parallel reduction commutes with the rigid expansion and since having a positive linear approximant in normal form implies being a $\beta$-normal form, we are beginning to grasp what is the needed condition on the Taylor expansion to characterize $\beta$-normalization. If there exists $\tm \in \rexp{M} $ such that the normal form of $\tm$ is positive then $M$ should be $\beta$-normalizable. In a more formal way:

\begin{proposition} \label{impbeta}
Let $M\in \ndTerms$. If there exists $\tm \in \rexp{M} $ and $ n \in \mathbb{N} $ such that $ L_{r}^{n}(\tm)\in D_{NF}^{+} $ then $L^{n}(M) $ is a $ \beta$-normal form.
\end{proposition}
\begin{proof}
 Let $s\in \rexp{M}$ such that there exists $NF(\tm) \in \resTerms^{+}$. Then there exists $n\in \mathbb{N}$ such that $L_{r}^{n} = NF(\tm). $ By Lemma \ref{forcingL} $L_{r}^{n}(\tm) \in \rexp{L^{n}(M)}$. Then  $L^{n}(M)$ is the $\beta$-normal form of $M$.
\end{proof}

\begin{proposition} \label{pbnf}
Let $M \in \ndTerms$. If $M$ is $\beta$-normalizable then there exists $\tm \in \rexp{M}$ such that $ NF(\tm) \in \rigidTerms^{+}$.
\end{proposition}
\begin{proof}
We consider the $N$, the normal form of $M$. By definition a $ \tmtwo \in \rexp{N}$ is a resource normal form. We choose a positive $\tmtwo$. Then, by Lemma \ref{antired1}, we can follow the anti-reduction and we can conclude.

\end{proof}

 \begin{theorem}
 \label{bnorm}
Let $M\in \ndTerms$. the following statements are equivalent:\\
(i) there exists $\tm \in \rexp{M}$ such that $ NF(\tm)\in \rigidTerms^{+} ;$\\
(ii) there exist $\tm \in \rexp{M} $ and $ m \in \mathbb{N} $ such that $ L_{r}^{m}(\tm) \in D_{NF}^{+}$;\\
(iii) $ L^{m}(M) $ is a $\beta$-normal form;\\
(iv) M is $\beta $-normalizable.\\

\end{theorem}

\begin{proof}
$(i)\Rightarrow (ii) $ is Lemma \ref{standL}. $(ii)\Rightarrow (iii) $ is Proposition $\ref{impbeta}. $ $(iii) \Rightarrow (iv) $ is trivial. Finally, $(iv) \Rightarrow (i)$ is Proposition $ \ref{pbnf}.$
\end{proof}

\section{Taylor expansion and rigid approximation}

In order to define the Taylor expansion of $\lambda$-terms we need to introduce another auxiliary language, the standard resource calculus. This calculus is just like the rigid resource calculus, where we replace lists with multisets.

	We  define the set of \definitive{resource terms} $\resTerms$
	and the set of \definitive{resource monomials} $\resMonomials$
	by mutual induction  as follows: 
\[
	\resTerms \ni s,t,u \recdef
		x \mid \lambda x. s \mid \rappl{s}{\bar{t}}		
	\qquad\qquad
	\resMonomials \ni \bar s,\bar t,\bar u \recdef
		[] \mid [s] \cdot \bar{t}
\]
We write $[s_1,\dotsc,s_n]$ for $[s_1]\cdot\ldots\cdot[s_n]\cdot[]$.
Monomials are then considered up to permutations and resource terms up to renaming of bound variables.
We call \definitive{resource expressions} the elements of $\resExprs=\resTerms\cup\resMonomials$.
For any resource expression $e$, we write $n_x(e)$ for the number of occurrences of variable $x$ in $e$.
If $A$ is a set, we write $\sums A$ for the set of finite formal sums of elements of $A$,
or equivalently the set of finite linear combinations of elements of $A$ with coefficients in $\N$.
We extend  the syntactical constructs of the resource calculus to finite sums of resource expressions by linearity:
e.g., $[s+t]\cdot \bar u=s\cdot\bar u+t\cdot\bar u$:

\begin{definition}
Let $ \sigma\in \mathbb{N}[\Delta_{\oplus}^{(!)}] $. We call $ \sigma $ a \emph{ finite term}. By linearity we extend  the syntactical constructs of the resource calculus to finite terms:
   \begin{itemize}
    \item if $ \sigma = \sum_{i=1}^{n} s_{i} $ and $ x \in \mathcal{V} $  we set $ \lambda x. \sigma = \sum_{i=1}^{n} \lambda x. s_{i} ; $

    \item if $ \sigma = \sum_{i=1}^{n} s_{i} $ and $ \bar{\tau} = \sum_{j=1}^{n} \bar{t}_{j} $ we set $ \langle \sigma \rangle \bar{\tau} = \sum_{i=1}^{n} \sum_{j=1}^{n} \langle s_{i} \rangle \bar{t}_{j} . $
   
    \item if $\sigma =  \sum_{i=1}^{n} s_{i}$ and $ \bar{\tau} = \sum_{j=1}^{n} \bar{t}_{j} $ we set $ [\sigma] \cdot \tau = \sum_{i=1}^{n} \sum_{j=1}^{n} [s_{i}] \cdot \bar{t}_{j}  $.
  \end{itemize}
  
\end{definition}

Let $ \Gamma \subseteq \Delta $. We set $ \sigma \subseteq \Gamma $ when $ supp(\sigma) \subseteq \Gamma. $

\begin{definition}
  Let $ e \in \resExprs $, $ \bar{u}= [u_{1},\dotsc,u_{n}] \in \resTerms^{!} $ and $ x\in \mathcal{V}$.
	We  define  the \definitive{$n$-linear substitution} $ \nsubst ex{\bar u}$
	of $ \bar{u} $ for $ x $ in $ e $ as follows:
	\[
		\nsubst ex{\bar u} = \begin{cases}
			\sum\limits_{\sigma\in \mathfrak{S}_{n}} e[ u_{\sigma(1)}/x_{1},\dotsc, u_{\sigma(n)}/x_{n} ] & \text{ if } n_{x}(e) = n
			\\
			0 & \text{ otherwise}
		 \end{cases}
	\]
	where $x_1,\dotsc,x_{n_x(e)}$ enumerate the occurrences of $x$ in $e$.
\end{definition}

\begin{figure}[t]
\begin{center}
	\begin{prooftree}\infer0{\rappl{\lambda x.s}{\bar t}\to_\partial \nsubst sx{\bar t}}\end{prooftree}
	\qquad
	\bigskip
	\begin{prooftree}\hypo{s\to_\partial \sigma'}\infer1{\lambda x.s\to_\partial\lambda x.\sigma'}\end{prooftree}
	\hfill
	\begin{prooftree}\hypo{s\to_\partial \sigma'}\infer1{\rappl s{\bar t}\to_\partial\rappl{\sigma'}{\bar t}}\end{prooftree}
	\hfill
	\begin{prooftree}\hypo{\bar s\to_\partial \bar \sigma'}\infer1{\rappl{t}{\bar s}\to_\partial \rappl{t}{\bar \sigma'}}\end{prooftree}
	\hfill
	\begin{prooftree}\hypo{s\to_\partial \sigma'}\infer1{[s]\cdot{\bar t}\to_\partial[\sigma']\cdot{\bar t}}\end{prooftree}
	\qquad
\begin{prooftree}\hypo{\bar s\to_\partial \bar \sigma'}\infer1{[t]\cdot\bar s\to_\partial [t]\cdot \bar \sigma'}\end{prooftree}
\end{center}
\caption{Reduction rules of the resource calculus with sums}
\label{fig:resred}
\end{figure}

The reduction of the resource calculus is the relation from 
resource expressions to finite formal sums of resource expressions 
induced by the rules of Figure \ref{fig:resred}.

We define the representation relation $\lhd \subseteq D^{(!)}\times \resExprs$ by the following rules:
\[
{
\begin{prooftree}
\infer0[]{ x \lhd x  }
\end{prooftree}
}
\quad
{
\begin{prooftree}
\hypo{  a \lhd s }
\infer1[]{  \lambda x. a \lhd \lambda x. s }
\end{prooftree} 
}
\quad
{
\begin{prooftree}
\hypo{  c\lhd s}
\hypo{ \ \vec{d}\lhd \bar{t}}
\infer2[]{  \langle c \rangle \vec{d} \lhd \langle s \rangle \bar{t} }
\end{prooftree}
}
\]\[{
\begin{prooftree}
\hypo{\sigma \in \mathfrak{S}_{n}\quad a_{1}\lhd t_{\sigma(1)}\quad\cdots\quad a_{n}\lhd t_{\sigma(n)}}
\infer1[]{  (a_{1},\dotsc, a_{n}) \lhd [t_{1},\dotsc, t_{n}]   }
\end{prooftree} 
\quad.
}
\]

We can extend the representation relation to linear combination of resource terms:  $a \lhd \sigma $ if there exists $ s \in supp(\sigma) $ such that $a \lhd s. $

We set the rigid expansion of $s$ as $ T_{r}(s) = \{ a\in D^{(!)} \mid a \lhd s \}.$ The rigid expansion of a resource term is an equivalence class of rigid resource terms.

We set \[{\nsubst {T(M)}{x}{T(N)^{!}}} = \bigcup\limits_{s \in T(M), \bar{t} \in T(N)^{!}} supp(\nsubst {s}{x}{\bar{t}}) .  \]

\begin{lemma}[Substitution]\label{subres}
Let $ M, N \in \Lambda ,$ we have that $ T(\subst {M} {x}{N}) = \nsubst {T(M)}{x}{T(N)^{!}} .$

\end{lemma}
\begin{proof} We prove the Lemma by induction on the definition of $M[N/x]$. 

 if $ M = x $  , then $ M [N/x] = N$. Thus we have to prove that $ T(N) = \bigcup_{ \bar{t}\in T(N)^{!}}  supp( \partial_{x} x \cdot \bar{t} ).  $ 
By double inclusion, if $ t\in T ( N) $ , lets consider $[t] \in T( N)^{!} $.
By definition of $n$-linear substitution we have that $ t\in supp(\partial_{x} x \cdot [t]) $. 
We prove now the other inclusion. For $ t\in \bigcup_{ \bar{t}\in T(N)^{!}}  supp( \partial_{x} x \cdot \bar{t} ) $, there exists $\bar{t}\in T(N)^{!} $ 
such that $ t\in supp(\partial_{x} x \cdot \bar{t}) $. 
By definition of $n$-linear substitution, $ \bar{t} = [t] $.  Then $ t\in T (N) $. 

 If $ M = y $ with $ y\neq x $ the result follows immediately 
from the definition of $n$-linear substitution. \par
If $ M = \lambda y. M' $ with $ y = x $ the result follows immediately from the definition of $n$-linear substitution.

 If $ M = \lambda y. M' $ with $ y \neq x $, let 
$ p\in T ( \lambda y. M'[N/x]) $. By definition of the Taylor expansion  there exists 
a $ p'\in T (M' [N/x]) $ such that $ p = \lambda y. p' $. By IH we have that $ T (M' [N/x]) = \bigcup_{ s\in T(M'), \bar{t}\in T(N)^{!}}  supp(\partial_{x} s \cdot \bar{t} )  $. Thus there exists $ s\in T (M')$ and a $\bar{t}\in T(N)^{!} $ such that $ p'\in supp(\partial_{x} s \cdot \bar{t} ) $. 
Therefore, by definition of   $n$-linear substitution,  $ p\in supp(\partial_{x} \lambda y. s\cdot \bar{t}) $.

 Conversely, let $ p\in \bigcup_{ s\in T(\lambda y. M'), \bar{t}\in T(N)^{!}}  supp(\partial_{x} s \cdot \bar{t} ) $. We have that, by definition, $p$ is of the form $ \lambda y. p' $, for some resource term $ p'$. There exists then a $s\in T (\lambda y. M') $ and a $ \bar{t} \in T (N)^{!} $ such that $ p\in supp(\partial_{x} s \cdot \bar{t} ) $.
 By definition of $n$-linear substitution, by IH and by the fact that there exists $ s'\in T ( M')$ such that $ s = \lambda y. s' $, we have that $ p'\in  supp( \partial_{x} s' \cdot \bar{t} ) \subseteq \bigcup_{ s\in T(M'), \bar{t}\in T(N)^{!}}  supp(\partial_{x} s \cdot \bar{t} ) = T (M' [N/x])  $ . Then, by definition, we can conclude that $ p = \lambda y. p' \in T (\lambda y. M'[N/x]) $. \par
If $ M = PQ $ let $ p\in T(PQ [N/x]) = T( P [N/x] Q [N/x]) $. Then, by definition of the Taylor expansion, there exists $ p' \in T (P[N/x]) $ and  $  \bar{q}  \in T (Q[N/x])^{!} $ such that $ p = \langle p'\rangle \bar{q} $.
By IH we have that $ p' \in T (P [N/x]) = \bigcup_{ s\in T(P), \bar{t}\in T(N)^{!}}  supp(\partial_{x} s \cdot \bar{t} )  $ and that \[ \bar{q} \in T (Q [N/x])^{!} = \left( \bigcup_{ s\in T(Q), \bar{t}\in T(N)^{!}}  supp(\partial_{x} s \cdot \bar{t} ) \right)^{!} .\]
Let $ \bar{q} = [q_{1},..., q_{n}].$
Then there exists $s_{0} \in T ( P ) $, there exists $ \bar{t}_{0} \in T ( N)^{!} $ and, for all $ i\in \{1,..., n\} $, there exist $ s_{i} \in T ( Q) $ and $ \bar{t}_{i} \in T (N) $ such that 
$ p' \in supp(\partial_{x} s_{0} \cdot \bar{t}_{0} ) $, and $ q_{i} \in supp( \partial_{x} s_{i} \cdot \bar{t}_{i}) $.
Taking $ \bar{t}= \sum_{i = 0}^{n} \bar{t}_{i} $ and $ \bar{s} = [s_{1}, ..., s_{n}] $  we have, from the definition of $n$-linear substitution, that 
$ \langle p' \rangle \bar{q} \in supp( \partial_{x} \langle s_{0} \rangle \bar{s} \cdot \bar {t} ) $.\par
Conversely, let $ p\in \bigcup_{ s\in T( P Q), \bar{t}\in T(N)^{!}}  supp(\partial_{x} s \cdot \bar{t} ) $. 
Then there exists $ s' \in T ( P) $ and a $ \bar{q} \in T(Q)^{!} $ such that $ p \in supp( \partial_{x} \langle s' \rangle \bar{q} \cdot 
\bar{t}) $ . Moreover, since $ \bar{t} = [ t_{1}, ..., t_{n} ], $ 
by definition of $n$-linear substitution we have that \[ \partial_{x} \langle s' \rangle \bar{q} \cdot \bar{t} = \sum_{I_{0}, I_{1}  \in
 Partitions (\{ 1, ..., n \})} \partial_{x} \langle s'  \cdot \bar{t}_{I_{0}} \rangle \partial_{x} \bar{q} \cdot \bar{t}_{I_{1}} .\]  Hence $ p = \langle p' \rangle \bar{q}' $ for some $ p'\in 
 supp(\partial_{x} s' \cdot \bar{t}_{I_{0}}) $ and $ \bar{q}' \in supp( \partial_{x} \bar{q} \cdot \bar{t}_{I_{1}}) $. 
 From the IH we know that $ T (P [N/x]) = \bigcup_{ s\in T(P), \bar{t}\in T(N)^{!}}  supp(\partial_{x} s \cdot \bar{t} )  $ and that $  T (Q [N/x]) = \bigcup_{ s\in T(Q), \bar{t}\in T(N)^{!}}  supp(\partial_{x} s \cdot \bar{t} )  $.
 Thus $ p' \in T (P [ N /x ] ) $ and $ \bar{q}' \in T ( Q [ N / x] )^! $. By definition of  the Taylor
 expansion of an application  $ p = \langle p'\rangle \bar{q}' \in T ( P [ N/x] Q[N/x])$. 

\end{proof}

\section{Strong normalization } \label{strong}

For strong normalisation we switch form rigid approximation to Taylor expansion.\footnote{The rigid approximation fails confluence of the extended non-erasing reduction, that is at the heart of our proof. The failure of confluence is interesting, since depends completely on the rigidity of the calculus (Section \ref{fail}). }

\subsection{Non-erasing reduction}

\begin{definition}
We define $ {\to^{\neg e}} \subseteq \Lambda\times \Lambda$ by induction as follows:
 \begin{itemize}
  \item $(\lambda x. M)N \to^{\neg e} M[N/x]$ if $ x \in FV(M)$;
  \item $\lambda x. M \to^{\neg e} \lambda x. M'$ if $M \to^{\neg e} M';$
  \item $PQ \to^{\neg e} P'Q$ if $ P\to^{\neg e} P' $ 
  \item $PQ\to^{\neg e} PQ'$ if $ Q\to^{\neg e} Q'.$
  \end{itemize}
  \end{definition}

A $\lambda$-term $M$ of the shape $ (\lambda x. M)N$ with $ x \in FV(M)$ is called a redex. A normal form for the non-erasing reduction is a $\lambda$-term $M$ that does not have redexes as subterms. A $\lambda$-term $M$ is called normalizable if there exist $M_{1},..., M_{n}$ such that $ M = M_{0} \to^{\neg e} M_{1} ... \to^{\neg e} M_{n} = N$ with $N$ being a normal form for the non-erasing reduction. We trivially have that $ \to^{\neg e} \subseteq \to_{\beta}.$

To understand the meaning of our definition, we can consider some example of non erasing reductions: 

\begin{example}

Non-erasing reduction at work: 
\begin{itemize}
\item $(\lambda x. y) \Omega$ is not non-erasing normalizable, since the variable $x$ is not free in the term $y$. However the term is clearly $\beta$-normalizable;
\item $ (\lambda x. y) z $ is a non-erasing normal form. Clearly it is not a $\beta$-normal form, since it contains a $\beta$-redex. 
\end{itemize}

\end{example}

At this point we could hope that non-erasing reduction characterizes strong normalization. However this is not at all the case.

 Let $M = ((\lambda y. \lambda x. x x) z) \lambda x. x x.$ Then $M$ is by definition a non-erasing normal form, but it is not even $\beta$-normalizable: 

\[ M \to_{\beta} (\lambda x. x x) \lambda x. x x \]

And trivially $(\lambda x. x x) \lambda x. x x = \Omega$ is not $\beta$-normalizable. 
\\

To solve this problem, we follow the path of linear logic. As presented in \cite{regnier:sigma}, MELL proof-nets induces a new kind of reduction on $\lambda$-terms, the so-called $\sigma $-reduction.  $\lambda$-calculus syntax induces a strict and unnecessary order on redexes. The $\sigma $-rules then grant some commutations of redexes that "free" $\lambda$-terms from this purely syntactical constraints. 

We define $ \to_{\sigma} \in \Lambda \times \Lambda $ as the contextual extension of the following rule:

\begin{gather*}
((\lambda x. M) N ) P \to_{\sigma 1} ( \lambda x. M P ) N \text{ if } x \notin FV(P)\\
\end{gather*}

Then we set $ {\to_{\epsilon}} = { \to_{\beta} \cup \to_{\sigma}}$.

Secularly, we define the erasing reduction as follows:

\begin{definition}
We define $ {\to^{e}} \subseteq \Lambda\times \Lambda$ by induction as follows:
 \begin{itemize}
  \item $(\lambda x. M)N \to^{ e} M[N/x]$ if $ x \notin FV(M)$;
  \item $\lambda x. M \to^{e} \lambda x. M'$ if $M \to^{ e} M';$
  \item $PQ \to^{ e} P'Q$ if $ P\to^{ e} P' $ 
  \item $PQ\to^{ e} PQ'$ if $ Q\to^{ e} Q'.$
  \end{itemize}
  \end{definition}


\subsection{Taylor expansion and non-erasing reduction}

In order to achieve a strong normalisation Theorem we have to switch from rigid terms to standard resource terms. The problem with rigid terms is indeed their rigidity: if we extend the calculus with $ \sigma_1$ we get a non confluent calculus in a very ba sense (see section ). On the contrary, standard resource calculus does not fail confluence, thanks to its intrinsic "non-deterministic" nature (see section \ref{A}.)\footnote{This problem could have been solved also changing the syntax of rigid terms or switching to polyhadic calculus \cite{mazza:pol}. However we preferred to stick on the Taylor expansion of $ \lambda -$terms, since our work is inspired mostly from that framework.}

We firstly extend the notion of $ \epsilon-$reduction to the resource calculus: 

\begin{definition}
We define $ \to_{\partial \sigma} \in \resExprs \times \resExprs $ as the contextual extension of the following rule:

\[ \rappl{\rappl{\lambda x. s} {\bar{t}}} {\bar{q}} \to_{\partial \sigma 1} ( \rappl{\lambda x. \rappl {s} {\vec{q}} } {\bar{t}} \text{ if } x \notin FV(\bar{q}). \]


\end{definition}

Then we set $ \tope = { \to_{\partial} \cup \to_{\partial \sigma}}$.

\begin{lemma} \label{snc}
The reduction $\tope$ is strongly normalizing. 
\end{lemma}
\begin{proof}
Strong normalisation derives form the fact that both $\to_{r\epsilon}$ and $\to_{r\sigma}$(the height of terms is decreasing) are strongly normalisable and by a transposition of Lemma 3.4 of \cite{regnier:sigma} to resource terms.
\end{proof}

We extend the non erasing reduction to the resource calculus:

\begin{definition}
We define $ \to_{\partial}^{\neg e} \in \resExprs \times \resExprs $ as the contextual extension of the following rule:

\[  \rappl{\lambda x. s}{\bar{t}} \to_{\partial}^{\neg e}\text{ if } x \in FV(s) \]

\end{definition}

We set $ {\tope^{\neg e}} = {\to_{\partial}^{\neg e}} \cup  {\to_{\partial \sigma}}. $

We extend also the erasing reduction:

\begin{definition}
We define $ \to_{\partial}^{\neg e} \in \resExprs \times \resExprs $ as the contextual extension of the following rule:

\[  \rappl{\lambda x. s}{\bar{t}} \to_{\partial}^{e} \text{ if } x \notin FV(s) \]

\end{definition}

\begin{lemma} \label{switch}
Let $s,t,u \in \resTerms$. If $ s \to^{e}_{\partial} t$ and  $t\twoheadrightarrow_{\partial \epsilon }^{\neg e} u$ then there exists $ t' \in \resTerms$ such that $ s\twoheadrightarrow_{\partial \epsilon}^{\neg e} t'$ and $t' \to^{e}_{\partial} u$. 
\end{lemma}

\begin{proof}
See Section \ref{post}.
\end{proof}

 \begin{lemma} \label{snce}
The reduction $\tope^{\neg e}$ is strongly normalizing and confluent. 
\end{lemma}
\begin{proof}
The strong normalisation is a corollary of Lemma \ref{snc}. Confluence is proved in \ref{A}.

\end{proof}
 
 We write $NF_{\epsilon}^{\neg e}(s)$ for the unique non-easing $ \epsilon$-normal form of $s$ that is a finite term, possibly the zero sum.

\begin{lemma}\label{epresnf}
Let $ s \in \Delta $. If $ s \to_{\partial \epsilon }^{e} \sigma  \neq 0$ and  $ NF(\sigma)^{\neg e}_{\epsilon} = 0 $ then $ NF(s)^{\neg e}_{\epsilon} = 0 .$
\end{lemma}
\begin{proof}
By induction on the size of $ s .$ 

If $ s = \lambda x_1 \dots \lambda x_m \rappl{\rappl{\cdots\rappl{x}{\bar{q_1}}}{\cdots \bar{q_i} \cdots}}{\bar{q_n}} $ then \[\sigma = \lambda x_1 \dots \lambda x_m \rappl{\rappl{\cdots\rappl{x}{\bar{q_1}}}{\cdots \tau_i \cdots}}{\bar{q_n}} \] with $ \bar{q_i} \to_{\partial \epsilon }^{e} \tau_i. $ Since $ NF(\sigma)^{\neg e}_{\epsilon} = 0 $ then $ NF(\bar{q_j})^{\neg e}_{\epsilon} = 0 $ for some $ j \in \{1, \dots, n \}. $ If $ j = i $ we apply the IH and we conclude by linearity. If $ j \neq i $ we conclude by Lemma \ref{snce}.

If $ s = \rhred $ then we can focus on the case where the erasing step is performed on the head-redex, since the other cases follows the same structure of above. Since $ \sigma \neq 0 $ we have that $ q_0 = [] .$ Then $ \sigma = \lambda x_1 \dots \lambda x_m \rappl{\rappl{\cdots\rappl{p}{\bar{q_1}}}{\cdots \bar{q_i} \cdots}}{\bar{q_n}} $. Now we perform $ n $ steps of $\sigma$-reduction on $ s $ obtaining \[s' = \lambda x_1 \dots \lambda x_m \rappl{ \lambda x. \rappl{\rappl{\cdots\rappl{p}{\bar{q_1}}}{\cdots \bar{q_i} \cdots}}{\bar{q_n}}}{[]}. \]  Then we apply the IH and we get $  NF(s)^{\neg e}_{\epsilon} = 0 .$

\end{proof}

After having extended the reduction relation to resource term, we seek a connection between the $ \epsilon-$reduction over ordinary $\lambda$-terms and its resource counterpart.

\begin{lemma} [Subject expansion] \label{antired}
Let $M, N$ be any two $\lambda$-terms. If $M \to^{\neg e}_{\epsilon} N$ then for all $ t_0 \in T(N) $ there exists $s \in T(M)$ and $ t_1, \dots, t_n \in T(N) $ such that $ s \twoheadrightarrow^{\neg e +}_{\partial \epsilon} \sum_{i=0}^{n} t_i. $ If $ t_0 $ is positive, then $ s \twoheadrightarrow^{\neg e }_{\partial \epsilon} \sum_{i=0}^{n} t_i.  $

\end{lemma}

\begin{proof}
By induction on the definition of non-erasing $\epsilon$-reduction. The base case derives from Lemma \ref{subres}.

The interesting case is the application case. Let $ M = PQ  $ and $ N = PQ' $ with $ Q \tope^{\neg e} Q'  $. Let $ t_{0} \in T(PQ') .$  Then $ t_0 = \rappl{p_0}{\bar{q'_0}} $ with $ p_{0}  \in T(P)$ and $\bar{q'_0} \in T(Q')^{!} .$ By IH there exists $ \bar{q'_1}, \dots, \bar{q'_n}  \in T(Q')^{!} $ and $ \bar{q} \in T(Q)^{!} $ such that $ \bar{q}  \twoheadrightarrow^{\neg e +}_{\partial \epsilon} \sum_{i=0}^{n} \bar{q'}_i. $  Then we can apply the IH and conclude. IF $ t_0 \in \Delta^{+} $ then, in particular, $ \bar{q_0} \neq [] $ and we can strengthen the IH with $\bar{q}  \twoheadrightarrow^{\neg e }_{\partial \epsilon} \sum_{i=0}^{n} \bar{q'}_i$.

\end{proof}

\begin{lemma}[Subject reduction] \label{subred}
Let $M, N$ be any two $\lambda$-terms. If $M \to^{\neg e}_{\epsilon} N$ then there exists $ t_1, \dots, t_n  \in T(N) $ such that for all $s \in T(M)$ such that $ s \twoheadrightarrow^{\neg e +}_{\partial\epsilon} \sum_{i=1}^{n} t_i. $ If $ s  $ is positive, then $ s \twoheadrightarrow^{\neg e}_{\partial \epsilon} \sum_{i=1}^{n} t_i. $

\end{lemma}

\begin{proof}
By induction on the definition of $\epsilon$-reduction and by Lemma \ref{subres}. 

The interesting case is the application case. Let $ M = PQ  $ and $ N = PQ' $ with $ Q \tope^{\neg e} Q'  $. Let $ s \in T(PQ) .$  Then $ s = \rappl{p}{\bar{q}} $ with $ p  \in T(P)$ and $\bar{q} \in T(Q)^{!} .$ By IH there exists $ \bar{q'_1}, \dots, \bar{q'_n}  \in T(Q')^{!} $ such that $ \bar{q}  \twoheadrightarrow^{\neg e +}_{\partial \epsilon} \sum_{i=1}^{n} \bar{q'}_i. $  Then we can apply the IH and conclude. IF $ s\in \Delta^{+} $ then, in particular, $ \bar{q} \neq [] $ and we can strengthen the IH with $\bar{q}  \twoheadrightarrow^{\neg e }_{\partial \epsilon} \sum_{i=1}^{n} \bar{q'}_i$.
\end{proof}

\begin{lemma} \label{Taysn}
If $M$ is  normalizable through non-erasing $\epsilon$-reduction then there exists $s \in T(M)$ such that $NF(s)_{ \epsilon}^{\neg e} \cap {\resTerms}^{+} \neq \emptyset$. 

\end{lemma}

\begin{proof}
The result is a corollary of Lemma \ref{antired}. Since $M$ is non-erasing $\epsilon$-normalizable then there exists a $\lambda$-term $N$ that is its $\epsilon$-normal form. If we consider a reduction chain starting from $M$ and ending in $N$ such as $ M \to_{\epsilon}^{\neg e} ... \to_{\epsilon}^{\neg e} N$, by Lemma \ref{antired} for all $ t_0, t_1, \dots, t_n \in T(N) $ we can find an element $s \in T(M)$ such that $ s \twoheadrightarrow_{\partial\epsilon}^{\neg e} \sum_{i = 0}^{n} t_i  $. Then, by an easy inspection of the definitions, $NF(s)_{\epsilon}^{\neg e} = \sum_{i = 0}^{n} t_i$. If we choose a positive $t_0$ we can then conclude.

\end{proof}

\begin{lemma}
Let $s \in \resTerms^{+}, x\in \mathcal{V}, \bar{t} \in \resMonomials_{+}$. Then $ \nsubst{s}{x}{\bar{t}} \in \resTerms^{+}.$
\end{lemma}
\begin{proof}
By induction on the definition of $n$-linear substitution. 

\end{proof}

\begin{proposition} \label{pres}
Let $s \in \resTerms^{+}$. If $ s \to_{r} t$ then $t \in \resTerms^{+}.$
\end{proposition}
\begin{proof}
The base case follows from the former lemma and the inductive cases follow immediately from the IH.

\end{proof}

\begin{lemma} \label{fromsnetosn}
Let $ M \in \Lambda$ such that $ M  $ is strongly normalisable. Then $ M $ is strongly normalisable for the non-erasing $\epsilon $-reduction.
\end{lemma}
\begin{proof}
By absurd \textit{via} the corollary 3.5 of \cite{regnier:sigma}.
\end{proof}

\begin{definition}
We define a set $ S $ of $ \lambda$-terms by induction as follows: 

\begin{itemize}
 \item If $ M_{1}, \dots, M_{n} \in S $ then $ x M_{1} \dots M_{n} \in S$;
 \item if $ M \in S $ then $ \lambda
  x. M \in S ; $
  \item if $ M_1 \in S $  and $ \subst {M_0}{x} {M_{1}} \dots M_{n} \in S $ then $ (\lambda x. M_0) M_{1} \dots M_{n} \in S .$

\end{itemize}

\end{definition}

\begin{lemma}\label{S1}
If $ M \in S $ then $M $ is strongly normalisable. 
\end{lemma}
\begin{proof}
See \cite{raams:perp}.
\end{proof}

\begin{proposition}\label{S}
If there exists $s \in T(M) \cap \resTerms^{+}$ such that $ NF_{\epsilon}^{\neg e}(s) \neq 0$ then $M \in S . $ 

\end{proposition}

\begin{proof}

By induction on the size of $ s $. 

Let $ M = \lambda x_1 \dots \lambda x_m . x M_{1} \dots M_{n} $. Then the result follows immediately from the IH, since $ \to_{\partial \epsilon}^{\neg e} $ is confluent and strongly normalising (Lemma \ref{snce}). 

Let $M = \lambda x_{1}... \lambda x_{m}. (\lambda x. P) Q_{0} Q_{1}... Q_{n} $ and \[N = \lambda x_{1}... \lambda x_{m}. P[Q_{0}/ x] Q_{1}... Q_{n}.\] 

Let $ x \in FV(P) .$ By  Lemma \ref{subred} there exists $\tau$ such that $s \to^{\neg e}_{\partial \epsilon} \tau$ and $\tau = \sum_{i = 1}^{n} t_i $ with $ t_1, \dots, t_n \in T_{r}(N) $. $ \tau \neq 0$ by Lemma \ref{snce}, since $ NF_{\epsilon}^{\neg e}(s) \neq 0.$

Since \[s = \lambda x_1 ... \lambda x_{m}. \langle \dots \langle \lambda x. p \rangle \bar{q_0} \rangle \bar{q_1} \dots \rangle \bar{q_n} \] then there exists $ t = \lambda x_1 ... \lambda x_{m}. \langle \dots \nsubst {p}{x}{\bar{q_0}} \rangle \bar{q_1} \dots \rangle \bar{q_n} \in T(N) $ such $ s \to^{\neg e}_{\partial \epsilon} t + \sum_{i}^{n} t_{i} = \tau  $ for some $ t_{i} \in T(N) $. By IH, strong normalisation and confluence (Lemma \ref{snce}) we have that $  Q_{0}, N \in S . $ Then, by definition of $ S$, $ M \in S .$

If $x \notin FV(M) $ then $ s = \lambda x_1 ... \lambda x_{m}. \langle \dots \langle \lambda x. p \rangle \bar{q_0} \rangle \bar{q_1} \dots \rangle \bar{q_n} $. We take $ s' =   \lambda x_1 ... \lambda x_{m}. \langle \dots \langle \lambda x. p \rangle [] \rangle \bar{q_1} \dots \rangle \bar{q_n} $. 

Then $ s' \to^{e}_{\partial\epsilon} t $, with $ t =  \lambda x_1 ... \lambda x_{m}. \langle \dots \langle p  \rangle \bar{q_1} \dots \rangle \bar{q_n} \in \rexp{N} $. By Lemmas \ref{snce} and \ref{epresnf} we can conclude, since $s(t) < s(s) $ and by IH $  Q_{0}, N \in S $.

\end{proof}

\begin{theorem} \label{snorm}
Let $M \in \Lambda$. The following statements are equivalent:\\
(i) There exists $s \in T(M)$ such that $ NF_{\epsilon}^{\neg e}(s) \subseteq {\resTerms}^{+}$; \\
(ii) There exists $s \in T(M) \cap \resTerms^{+}$ such that $ NF_{\epsilon}^{\neg e}(s) \neq 0$; \\
(iii) $ M \in S $;\\
(iv) $M$ is strongly normalizable; \\
(v)$ M $ is non-erasing $ \epsilon $-normalisable.

\end{theorem}

\begin{proof}
$(i) \Rightarrow (ii) $ is a corollary of Propostion $ \ref{pres} $. $(ii) \Rightarrow (iii)$ is Proposition \ref{S}. $(iii) \Rightarrow (iv)$ is Lemma \ref{S1}. $ (iv) \Rightarrow (v) $ is Lemma \ref{fromsnetosn}. Finally, $(v) \Rightarrow (i)  $ derives from Lemma \ref{Taysn}.
\end{proof}

\subsection{Conservation Theorem for the $\lambda I$-calculus} \label{cons}

As corollary of Theorem \ref{snorm} we get Theorem 9.1.5 of \cite{bar:lam}. We define the set of $ \lambda I$-terms by induction as follows:

\[
	\Lambda I \ni M, N ::=
		x \mid \lambda x. M \text{ if } x \in FV(M) \mid MN \] 
		
In particular we have that $  \Lambda I \subset \Lambda.$

\begin{theorem}
Let $ M  $ be a $\lambda I $-term. Then $ M $ is normalisable iff $ M $ is strongly normalisable.
\end{theorem}		
\begin{proof}

$(\Rightarrow )$ By Theorem \ref{bnorm}.  Since $ M $ is a $\lambda I $-term, trivially $ M $ is non-erasing $\epsilon $-normalisable, by observing that a $ \beta$-normal form is also a non-erasing $ \epsilon$-normal form. Then $ M $ is strongly normalisable by Theorem \ref{snorm}.
$(\Leftarrow)$ Trivial.
\end{proof}

\section{Technicalities}

\subsection{Failure of confluence for rigid terms} \label{fail}

If we set $ {\to_{r \epsilon}} =  { \to_r \cup \to_{\sigma_1}} $ we get the following counter example to the confluence of $ {\to _{r \epsilon}} :$

\[ s = \rappl{\lambda y. \rappl{\rappl{\lambda x. \rappl{x}{(x)}}{(y,y)}}{(y)}}{(\lambda f. \rappl{z}{(f)},\lambda f. \rappl{z}{()}, \lambda f. \rappl{z}{()})} \] 

Then 

\[ s \to_{\sigma 1} s' =  \rappl{\lambda y. \rappl{\lambda x. \rappl{ \rappl{x}{(x)}}{(y)}}{(y,y)}}{(\lambda f. \rappl{z}{(f)},\lambda f. \rappl{z}{()}, \lambda f. \rappl{z}{()})}  \]

Then 

\[ s' \to_{r} s'' = \rappl{\lambda x. \rappl{ \rappl{x}{(x)}}{(\lambda f. \rappl{z}{(f)})}}{(\lambda f. \rappl{z}{()},\lambda f. \rappl{z}{()})}  \]

Then

\[ s'' \to_{r} s''' = \rappl{ \rappl{\lambda f. \rappl{z}{()}}{(\lambda f. \rappl{z}{()})}}{(\lambda f. \rappl{z}{(f)})} \]

and $ s''' \to_{r} 0 .$

But if one performs the $r$-reduction step before, $ NF(s) \neq 0 :$

\[ s \to_{r} s' =  \rappl{\rappl{\lambda x. \rappl{x}{(x)}}{(\lambda f. \rappl{z}{(f)},\lambda f. \rappl{z}{()})}}{(\lambda f. \rappl{z}{()})} \]. 

Then

\[ s' \to_{\sigma 1} s'' = \rappl{\lambda x. \rappl{ \rappl{x}{(x)}}{(\lambda f. \rappl{z}{()})}} {(\lambda f. \rappl{z}{(f)},\lambda f. \rappl{z}{()})} \]

And then in 2 steps of $r$-reduction one arrives to a non zero normal form.

The failure of confluence is due to the rigidity of the calculus, in the sense that the substitution does not perceives that the free occurrences of $ x $ are changing place after a step of $ \sigma -$reduction. 

This form of confluence failure is particularly bad because a term can have a zero and a non zero normal form. In this way all the approximation results are lost.

\subsection{Confluence} \label{A}

For any resource expression $e$, we write $n_x(e)$ for the number of occurrences of variable $x$ in $e$.
\begin{definition}
  Let $ e \in \resExprs $, $ \bar{u}= [u_{1},\dotsc,u_{n}] \in \resTerms^{!} $ and $ x\in \mathcal{V}$.
	We  define  the \definitive{$n$-linear substitution} $ \nsubst ex{\bar u}$
	of $ \bar{u} $ for $ x $ in $ e $ as follows:
	\[
		\nsubst ex{\bar u} = \begin{cases}
			\sum\limits_{S\in \mathfrak{S}_{n}} e[ u_{S(1)}/x_{1},\dotsc, u_{S(n)}/x_{n} ] & \text{ if } n_{x}(e) = n
			\\
			0 & \text{ otherwise}
		 \end{cases}
	\]
	where $x_1,\dotsc,x_{n_x(e)}$ enumerate the occurrences of $x$ in $e$.
\end{definition}
In other words, 
if $ \bar{u}= [u_{1},\dotsc,u_{n}] \in \Delta^{!} $ and $ x\in \mathcal{V}$ then:
\begin{gather*}
	\partial_{x} y \cdot \bar{u}=
	\begin{cases}   y &  \text{if} \  y \neq x   \text{ and }  n = 0\\
	u_{1} & \text{if }   y = x  \text{ and }  n= 1  \\ 0 & \text{otherwise}\end{cases}
	\\
	\partial_{x} {\lambda y. s}\cdot \bar{u} =
	\lambda y. ( \partial_{x} s \cdot \bar{u})\\
	\begin{aligned}
	\nsubst {\langle s \rangle \bar{t} } {x}  {\bar{u}} &=
	\sum_{(I_{0}, I_{1}) \text{ partition of }  \{1,\dotsc,n\} } \rappl { \nsubst  {s}{x} {\bar{u}_{I_{0}}} } {\nsubst  {\bar{t}} {x} {\bar{u}_{I_{1}} }}
	\\
	\partial_{x} [t_{1},\dotsc, t_{k}]\cdot \bar{u} &= 
	\sum_{(I_{1},\dotsc, I_{k}) \text{ partition of }  \{1,\dotsc,n\} } [\partial_{x} t_{1} \cdot \bar{u}_{I_{1}},\dotsc, \partial_{x} t_{n} \cdot \bar{u}_{I_{k}} ]
	\end{aligned}
\end{gather*}
where we write $\bar{u}_{\{i_1,\dots,i_k\}}:= [u_{i_1},\dots,u_{i_k}]$
whenever $1\le i_1<\dots<i_k\le n$.\footnote{
	To be precise, we say $(I_{1},.., I_{k})$ is a partition of a set $X$
	if the $I_j$'s are (possibly empty) pairwise disjoint subsets of 
	$X$ and $X=\bigcup_j I_j$.
	This data is equivalent to a function $\{1,\dotsc,n\}\to\{1,\dotsc,k\}$.
}

To prove the confluence of $ \to^{ \neg e}_{r \epsilon} $ (Lemma \ref{snce}) we use the standard technique of \cite{bar:lam}, defining a parallel non-erasing $ \epsilon $-reduction. 

\begin{definition}

We define $ \pare \subseteq \resExprs \times \resExprs $ by induction as follows:

\begin{itemize}
\item $ s \pare s ; $
\item If $  s \pare s' $ then $ \lambda x. s \pare \lambda x. s' ;$
\item If $ s \pare s' $ and $ \bar{t} \pare \bar{t'} $ then $ \rappl{s}{\bar{t}} \pare \rappl{s'}{\bar{t'}} ; $
\item  If $ s \pare s' $ and $ \bar{t} \pare \bar{t'} $ and $ x \in FV(s) $  then $ \rappl{\lambda x. s}{\bar{t}} \pare \nsubst {s'}{x}{\bar{t'}}; $
\item  If $ s \pare s' $, $ \bar{t} \pare \bar{t'} $, $\bar{q} \pare \bar{t'}  $ and $ x \notin FV(\bar{q}) $ then $ \rappl{\rappl{\lambda x . s}{\bar{t}}}{\bar{q}} \pare \rappl{\lambda x. \rappl{s'}{\bar{q'}}}{\bar{t'}}. $
\item if $ t_i \pare t'_i $ for $i \in \{ 1, \dots, n \} $ then $ [t_{1}, \dots, t_{n}] \pare [t'_{1}, \dots, t'_{n}]. $
\end{itemize}

\end{definition}

We extend the reduction defined above to finite terms by linearity: $ \sigma \pare \sigma' $ if there exists $ s \in supp(\sigma) $ and $ s' \in supp(\sigma') $ such that $ s\pare s' .$

\begin{lemma}\label{sconfl}
Let $x \in FV(s)$. If $ s \pare s' $ and $ \bar{t} \pare \bar{t'} $ then $\nsubst {s}{x}{\bar{t}} \pare \nsubst {s'}{x}{\bar{t'}}. $

\end{lemma}

\begin{proof}
The proof is by induction on the definition of $ s \pare s' . $  We notice that if $ \bar{t} \pare \bar{t'} $ then the two lists have the same size. Thus, by definition of $n $-linear substitution $ \nsubst {s}{x}{\bar{t}} = 0$ iff $ \nsubst {s'}{x}{\bar{t'}} = 0.$ Then we can focus on the case that $\nsubst {s}{x}{\bar{t}} \neq 0. $

\begin{enumerate}
\item If $ s' = s $ then the proof is by induction on the structure of $ \nsubst {s}{x}{\bar{t}} $:

\begin{itemize}
\item If $ s = x $ Then $\nsubst {s}{x}{[t_1]} = t_1 $ and $\nsubst {s}{x}{[t'_1]} = t'_1  $. By hypothesis $ t_1 \pare t'_1 $. Then we can conclude;
\item If $ s = \lambda x. s' $ then $ \nsubst {s}{x}{\bar{t}} = \lambda x. \nsubst {s'}{x}{\bar{t}}.  $ The result derives immediately from the IH; 
\item If $ s = \rappl{p}{\vec{q}} $ then $\nsubst {s}{x}{\bar{t}} = \nsubsta{p}{q}{x}{t}  $ with $ \bar{t} = \bar{t}_{I_0} \cdot \bar{t}_{I_1} $ for all $ ({I_0}, {I_1}) $ partitions of $\{ 1, \dots, n \} $. Then we apply the IH and conclude.
\item If $ s = [q_1, \dots, q_n] $ then $ \nsubst {s}{x}{\bar{t}} = [\nsubst {q_1}{x}{\bar{t}_1}, \dots,\nsubst {q_n}{x}{\bar{t}_n} ]$ with $ \bar{t} = \bar{t}_1 \cdots \bar{t}_n $. Then we can apply the IH and conclude.

\end{itemize}

\item If $ s' = \lambda x. p' $ the result derives immediately from the IH.

\item Let $ s' = \rappl{p'}{\vec{q'}} $ with $ s =  \rappl{p}{\vec{q}}$ with $ p \pare p' $ and $ \vec{q} \pare \vec{q'} $. We have that \[ \nsubst {s}{x}{\bar{t}} = \nsubsta{p}{q}{x}{t}\] and that \[ \nsubst {s'}{x}{\bar{t'}} = \nsubsta{p'}{q'}{x}{t'} . \]

 Then, by IH $ \nsubst {p}{x}{\bar{t}_{I_0}} \pare \nsubst {p'}{x}{\bar{t'}_{I_0}} $ and $\nsubst {\vec{q}}{x}{\vec{t}_{I_1}} \pare  \nsubst {\vec{q'}}{x}{\vec{t'}_{I_1}}$ for some $  ({I_0}, {I_1}) $ partitions of $\{ 1, \dots, n \} $. Then we can apply the IH and conclude.

\item Let $ s' = \nsubst {p'}{y}{\bar{q'}} $ with $ s = \rappl{\lambda y. p}{\bar{q}}. $ Then \[ \nsubst {s}{x}{\bar{t}} = \nsubsta{\lambda y. p}{q}{x}{t}.\] 

By IH \[ \rappl{\nsubst {\lambda y. p}{x}{\bar{t}_{I_0}}}{\nsubst {\bar{q}}{x}{\bar{t}_{I_1}}} \pare \rappl{\nsubst {\lambda y. p'}{x}{\bar{t'}_{I_0}}}{\nsubst {\bar{q'}}{x}{\bar{t'}_{I_1}}} \pare \nsubst{(\nsubst { p'}{x}{\bar{t'}_{I_0}})} {y}{\nsubst {\bar{q'}}{x}{\bar{t'}_{I_1}}}\] and by an inspection of the definition $\nsubst{(\nsubst { p'}{x}{\bar{t'}_{I_0}})} {y}{\nsubst {\bar{q'}}{x}{\bar{t'}_{I_1}}} \subseteq \nsubst{(\nsubst{p'}{y}{\bar{q'}})} {x} {\bar{t'}}. $

\item Let $ s' = \rappl{\lambda y. \rappl{p'}{\bar{u'}}}{\bar{q'}} $ with $ s = \rappl{\rappl{\lambda y . p}{\bar{q}}}{\bar{u}} .  $ 

Then \[ \nsubst {s}{x}{\bar{t}} = \nsubsta {\rappl{\lambda x. p}{\bar{q}}}{u}{x}{t} \] and \[ \nsubst {s'}{x}{\bar{t'}} =  \nsubsta {\lambda y. \rappl{p'}{\bar{u'}}}{q'}{x}{t'} \] 

Then \[ \nsubst{\rappl{\lambda y. p}{\bar{q}}}{x}{t_{I_0}} = \nsubstaa{\lambda y. p}{q}{x}{t_{I_0}} \]

And 

 \[ \nsubst{\lambda y. \rappl{ p}{\bar{u}}}{x}{t_{I_0}} = \lambda y. \nsubstaa{p}{u}{x}{t_{I_0}} \]
 
 Then by IH there exists $  ({J_0}, {J_1}), ({J'_0}, {J'_1}) $ partitions of $\{ 1, \dots, n \} $ such that $\nsubst {\bar{q}}{x}{{t_{J_0}}_{J'_1}} \pare \nsubst {\bar{q'}}{x}{{t'_{J_0}}_{J'_1}} $ and $ \nsubst {p}{x}{{t_{J_0}}_{J'_0}} \pare \nsubst {p'}{x}{{t'_{J_0}}_{J'_0}}$ and $ \nsubst{\bar{u}}{x}{t_{J_1}} \pare  \nsubst{\bar{u'}}{x}{t'_{J_1}}$.
 
 Hence \[\rappl{\nsubst{\rappl{\lambda y. p}{\bar{q}}}{x}{\bar{t}_{J_0}}}{\bar{u}_{J_1}} \pare \rappl{\rappl{ \lambda y. \nsubst {p'}{x}{{t'_{J_0}}_{J'_0}}}{\nsubst{\bar{u'}}{x}{t'_{J_1}}}}{\nsubst {\bar{q'}}{x}{{t'_{J_0}}_{J'_1}}}  \] 
 
 and 
 
 \[\rappl{\rappl{ \lambda y. \nsubst {p'}{x}{{t'_{J_0}}_{J'_0}}}{\nsubst{\bar{u'}}{x}{t'_{J_1}}}}{\nsubst {\bar{q'}}{x}{{t'_{J_0}}_{J'_1}}} \subseteq \] \[ \nsubsta {\lambda y. \rappl{p'}{\bar{u'}}}{q'}{x}{t'}  \]
 
 \item the multiset case is similar to the linear application case.
\end{enumerate}

\end{proof}

\begin{lemma}
$ \pare$ is confluent.
\end{lemma}

\begin{proof}
By induction on $ s \pare s_1 $ we prove that for all $ s_2 \in D$ such that $ s \pare s_2  $ we sow that there exists $ t $ such that $ s_1 \pare t $ and $ s_2 \pare t. $

\begin{enumerate}
\item if $ s_2 = s .$ Then take $  t = s_2 $.
\item if $ s_1 =   \rappl{\lambda x. \rappl{p'}{\bar{u'}}}{\bar{q'}} $ with $ s = \rappl{\rappl{\lambda x . p}{\bar{q}}}{\bar{u}} .  $ We have 2 possible cases:

\begin{itemize}
\item if $ s_2 =  \rappl{\rappl{\lambda x . p''}{\bar{q''}}}{\bar{u''}} $ we perform a step of sigma reduction and then we apply the IH and we conclude;
\item if $ s_2 =\rappl { \nsubst {p} {x} {\bar{q}}} {\bar{u}}$ then we apply the IH and the former lemma to conclude.
\end{itemize}

\item if $ s_1 = \nsubst{p'}{x}{\bar{q'}} $ with $ s = \rappl{\lambda x. p}{\bar{q}} $ we have 2 possible cases: 

\begin{itemize}
\item if $ s_2 = \nsubst{p''}{x}{\bar{q''}} $ with $ p \pare p'' $ and $ \bar{q} \pare \bar{q''} $. Then we apply Lemma \ref{sconfl} and conclude;
\item if $ s_2 =  \rappl{\lambda x. p''}{\bar{q''}} $ then by IH and Lemma \ref{sconfl} we can conclude; 
\end{itemize}

The other cases derives from a lengthy by completely standard induction, \textit{via} Lemma \ref{sconfl}. 

\end{enumerate}
\end{proof}

\begin{lemma}
$ \to^{* \neg e}_{\partial \epsilon} $ is the transitive closure of $ \pare .$
\end{lemma}
\begin{proof}
Easy inspection of the definitions.
\end{proof}

\begin{lemma}
$ \tope$ is confluent.
\end{lemma}
\begin{proof}
Observe that $ \tope^{\neg e} \subseteq \pare \subseteq \to^{* \neg e}_{\partial \epsilon}. $ Then apply Lemma 3.2.2 of \cite{bar:lam}.
\end{proof}

\subsection{Postponement} \label{post}

\begin{lemma} 
Let $s,t,u \in \rigidTerms$ . If $ s \to^{e}_{r} t$ and  $t\twoheadrightarrow_{r \epsilon }^{\neg e} u$ then there exists $ t' \in \rigidTerms$ such that $ s\twoheadrightarrow_{r \epsilon}^{\neg e} t'$ and $t' \to^{e}_{r} u$. 
\end{lemma}

\begin{proof}
By induction on length of $s \twoheadrightarrow^{\neg e}_{r \epsilon} t$.  Let $ l = 1 $.

If $ s = \rappl{\lambda x. p} {()} $ then we can immediately conclude by the following diagram:

\begin{center}
 
 \begin{tikzpicture}

  \matrix (m) [matrix of math nodes,row sep=3em,column sep=4em,minimum width=2em]
  {
     \rappl{\lambda x. p} {()}  & p \\
     \rappl{\lambda x. u} {()} & u \\};
  \path[-stealth]
    (m-1-1) edge [] node [left] {$\neg e$} (m-2-1)
            edge [] node [below] {$e$} (m-1-2)
    (m-2-1.east|-m-2-2) edge node [below] {$ e$}
            node [above] {}  (m-2-2)
    (m-1-2) edge [] node [right] {$\neg e$} (m-2-2)
           ;
\end{tikzpicture}
 \end{center}
 
 If $ s = \lambda x. p $ the result derives immediately by IH. 
 
 If $ s = \rappl{p}{\vec{q}} $ there are two possible cases: 
 
(i) $ p \to^{e}_{r\epsilon} p' $ and $ t = \rappl{p'}{\vec{q}}; $

(ii)$ \vec{q} \to^{e}_{r\epsilon} \vec{q'} $ and $ t = \rappl{p}{\vec{q'}}. $

If

 \[  \rappl{p'} {\vec{q}} \to^{\neg e}_{r\epsilon} u  \] 
 
 is an internal step, $ i.e. $ $ u = \rappl{u'}{\vec{q'}} $ (resp. $ u = \rappl{p'}{\vec{u'}} $) with $ p' \to^{\neg e}_{r\epsilon} u $ (resp. $ \vec{q}' \to^{\neg e}_{r\epsilon} \vec{u'} $ then the factorization 
is given directly by IH.
 
Otherwise we proceed by cases. If $ p' = \lambda x. v $ then $ u = \subst {v}{x}{\vec{q}} $. Then $ p $ has to be of the shape $ \rappl {\lambda y. \lambda x. v} {()} $. We can then conclude by the following diagram: 

 \begin{center}
 
 \begin{tikzpicture}

  \matrix (m) [matrix of math nodes,row sep=3em,column sep=4em,minimum width=3em]
  {
     \rappl{\rappl {\lambda y. \lambda x. v} {()}} {\vec{q}}  & \rappl {\lambda x. v} {\vec{q}} \\
     \rappl{\lambda y. \rappl {\lambda x. v} {\vec{q}}} {()} & u\\
                       \rappl{\lambda y. u} {()}   \\
     };
  \path[-stealth]
    (m-1-1) edge [] node [left] {$\sigma 1$} (m-2-1)
            edge [] node [below] {$e$} (m-1-2)
    (m-3-1.east|-m-3-1) edge node [below] {$ e$}
            node [above] {}  (m-2-2)
           (m-1-2) edge [] node [right] {$\neg e$} (m-2-2)
        (m-2-1)  edge [] node [left] {$\neg e$} (m-3-1)

    (m-1-2) edge [] node [right] {$\neg e$} (m-2-2)
           ;
\end{tikzpicture}
 \end{center}

The second base case is a redex for the $ \sigma_{1} $ rule:  
 
 \begin{center}
 
 \begin{tikzpicture}

  \matrix (m) [matrix of math nodes,row sep=3em,column sep=4em,minimum width=3em]
  {
     \rappl{\rappl {\lambda y. \rappl { \lambda x. v} {\vec{z}}} {()}} {\vec{q}}  & \rappl{\rappl {\lambda x. v} {\vec{z}}} {\vec{q}} \\
     \rappl{\lambda y. \rappl { \rappl{\lambda x. v}{\vec{z}}} {\vec{q}}} {()} & u\\
                       \rappl{\lambda y. \rappl {\lambda x. \rappl{ v}{\vec{q}}} {\vec{z}}} {()}    \\
     };
  \path[-stealth]
    (m-1-1) edge [] node [left] {$\sigma 1$} (m-2-1)
            edge [] node [below] {$e$} (m-1-2)
    (m-3-1.east|-m-3-1) edge node [below] {$ e$}
            node [above] {}  (m-2-2)
           (m-1-2) edge [] node [right] {$ \sigma 1 $} (m-2-2)
        (m-2-1)  edge [] node [left] {$\sigma 1 $} (m-3-1)

           ;
\end{tikzpicture}
 \end{center}

(ii) The proof follows a specular path to the proof of (i). 

If $ l = n + 1 $ the result follows immediately by IH. 
\end{proof}
 \bibliographystyle{alpha}
\bibliography{bib}
\end{document}